\documentclass[letterpaper,11pt]{article}
\usepackage{amsmath, amssymb}
\usepackage{color}
\usepackage{fullpage}
\usepackage[numbers]{natbib}
\usepackage[noend]{algorithmic}
\usepackage{tikz}
\usepackage{enumerate}
\usepackage{graphicx}
\usepackage{caption}
\usepackage{hyperref}
\usepackage{subcaption}
\usepackage{zerosum,csquotes}
\usepackage{enumitem,linegoal}
\usepackage[linesnumbered,ruled,vlined]{algorithm2e}
\usepackage{comment}	
\usepackage{ctable}					
\newtheorem{observation}{Observation}[section]
\usepackage{mathtools}
\usepackage[flushleft]{threeparttable}
\usepackage[normalem]{ulem}

\usepackage{footnote}
\makesavenoteenv{tabular}
\makesavenoteenv{table}

\usepackage[margin=1in]{geometry}

 \newtheorem{theorem}{Theorem}[section]
 \newtheorem{lemma}[theorem]{Lemma}

 \newtheorem{example}{Example}[section]

\makeatletter
\def\GrabProofArgument[#1]{ #1: \egroup\ignorespaces}
\def\proof{\noindent\textbf\bgroup Proof%
	\@ifnextchar[{\GrabProofArgument}{. \egroup\ignorespaces}}

\makeatother

\newcommand*\samethanks[1][\value{footnote}]{\footnotemark[#1]}

\newcounter{proccnt}

\newcommand{\konote}[1]{}

\title{Fair Allocation of Indivisible Goods to Asymmetric Agents}

\author{
	Alireza Farhadi \thanks{University of Maryland. Email: \texttt{\{farhadi,hajiagha,hadiyami\}@cs.umd.edu, sseddigh@umd.edu}}
	\thanks{Supported in part by NSF CAREER award CCF-1053605,  NSF BIGDATA grant IIS-1546108, NSF AF:Medium grant CCF-1161365, DARPA GRAPHS/AFOSR grant FA9550-12-1-0423, and another DARPA SIMPLEX grant.}
	\and Mohammad Ghodsi \thanks{Sharif University of Technology. Email: \texttt{ghodsi@sharif.edu, mseddighin@ce.sharif.edu}}
	\thanks{Institute for Research in Fundamental Sciences (IPM) – School of Computer Science.}
	\and MohammadTaghi HajiAghayi \samethanks[1] \samethanks[2]
	\and Sebastien Lahaie \thanks{ Google Research. Email: \texttt{slahaie@microsoft.com}}
	\and David Pennock \thanks { Microsoft Research. Email: \texttt{dpennock@microsoft.com}}
    \and Masoud Seddighin \samethanks[3]
	\and Saeed Seddighin \samethanks[1] \samethanks[2]
	\and Hadi Yami \samethanks[1] \samethanks[2]
}

\begin{document}
	\newcommand{\ignore}[1]{}
\renewcommand{\theenumi}{(\roman{enumi}).}
\renewcommand{\labelenumi}{\theenumi}
\sloppy

%
%

\newcommand{\agent}{a}
\newcommand{\agents}{\mathcal{N}}
\newcommand{\ite}{b}
\newcommand{\bundle}{B}
\newcommand{\items}{\mathcal{M}}
\newcommand{\valu}{V}
\newcommand{\MMS}{\mathsf{MMS}}
\newcommand{\WMMS}{\mathsf{WMMS}}
\newcommand{\shares}{E}
\newcommand{\share}{e}
\newcommand{\steinhausfirst}{Steinhaus}
\newcommand{\maxmin}{\textsf{Max-Min allocation}}
\newcommand{\santa}{\textsf{Santa Claus}}
\newcommand{\bezakovaallocating}{Bez{\'a}kov{\'a} and Dani}
\newcommand{\asadpourapproximation}{Asadpour and Saberi}
\newcommand{\chakrabartyallocating}{Chakrabarty, Chuzhoy and Khanna}
\newcommand{\bansalsanta}{Bansal and Sviridenko} 
\newcommand{\annamalaicombinatorial}{Annamalai, Kalaitzisy and Svensson} 
\newcommand{\procacciafirst}{Procaccia and Wang}
\newcommand{\todo}[2]{{ \color{blue} \textit{#1}:\color{magenta}{`#2'}}}
\newcommand{\etal}{\textit{et al.}}
\newcommand{\amanatidisapproximationful}{Amanatidis, Markakis, Nikzad, and Saberi}   

\maketitle

\thispagestyle{empty}

\begin{abstract}
We study fair allocation of indivisible goods to agents with \textit{unequal entitlements}. Fair allocation has been the subject of many studies in both divisible and indivisible settings. Our emphasis is on the case where the goods are indivisible and agents have unequal entitlements. This problem is a generalization of the work by \procacciafirst ~\cite{Procaccia:first} wherein the agents are assumed to be symmetric with respect to their entitlements. Although \procacciafirst\enspace show an almost fair (constant approximation) allocation exists in their setting, our main result is in sharp contrast to their observation. We show that, in some cases with $n$ agents, no allocation can guarantee better than $1/n$ approximation of a fair allocation when the entitlements are not necessarily equal. Furthermore, we devise a simple algorithm that ensures a $1/n$ approximation guarantee.

Our second result is for a restricted version of the problem where the valuation of every agent for each good is bounded by the total value he wishes to receive in a fair allocation. Although this assumption might seem w.l.o.g, we show it enables us to find a $1/2$ approximation fair allocation via a greedy algorithm. Finally, we run some experiments on real-world data and show that, in practice, a fair allocation is likely to exist. We also support our experiments by showing positive results for two stochastic variants of the problem, namely \textit{stochastic agents} and \textit{stochastic items}.
\end{abstract}

\section{Introduction}
In this work, we conduct a study of \textit{fairly} allocating \textit{indivisible goods} among $n$ agents with unequal claims on the goods. Fair allocation is a very fundamental problem that has received attention in both Computer Science and Economics. This problem dates back to 1948 when \steinhausfirst ~\cite{Steinhaus:first} introduced the \textit{cake cutting} problem as follows: given $n$ agents with different valuation functions for a cake, is it possible to divide the cake between them in such a way that every agent receives a piece whose value to him is at least $1/n$ of the whole cake? \steinhausfirst\enspace answered this question in the affirmative by proposing a simple and elegant algorithm which is called \textit{moving knife}. Although this problem admits a straightforward solution, several ramifications of the cake cutting problem have been studied since then, many of which have not been settled after decades~\cite{brams1996fair,robertson1998cake,chen2013truth,procaccia2013cake,edmonds2006cake,
cohler2011optimal,woeginger2007complexity,caragiannis2011towards,FGST17,brams17,babaioff17}. For instance, a natural generalization of the problem in which we discriminate the agents based on their entitlements is still open. In this problem, every agent claims an entitlement $\share_i$ to the cake such that $\sum \share_i = 1$, and the goal is to cut the cake into disproportional pieces and allocate them to the agents such that every agent $\agent_i$'s valuation for his piece is at least $\share_i$ fraction of his valuation for the entire cake. For two agents, Brams \etal ~\cite{brams2008proportional} showed that at least two cuts are necessary to divide the cake between the agents. Furthermore, Robertson \etal ~\cite{robertson1997extensions} proposed a modified version of cut and choose method to divide the cake between two agents with portions $\share_1,\share_2$, where $\share_1$ and $\share_2$ are real numbers. McAvaney, Robertson, and Web~\cite{mcavaney1992ramsey} considered the case when the entitlements are rational numbers. They used Ramsey partitions to show that when the entitlements are rational, one can make a proper division via $O(n^3)$ cuts.

Recently, a new line of research is focused on the fair allocation of indivisible goods. In contrast to the conventional cake cutting problem, in this problem instead of a heterogeneous cake, we have a set $\items$ of indivisible goods and we wish to distribute them among $n$ agents. Indeed, due to trivial counterexamples in this setting\footnote{For instance if there is only one item, at most one agent has a non-zero profit in any allocation.}, the previous guarantee, that is every agent should obtain $1/n$ of his valuation for all items from his allocated set, is impossible to deliver. To alleviate this problem, Budish~\cite{Budish:first} proposed a concept of fairness for the allocation of indivisible goods namely \textit{the maxmin share}. Suppose we ask an agent $\agent_i$ to divide the items between the agents in a way that \textit{he thinks} is fair to everybody. Of course, agent $\agent_i$ does not take into account other agents' valuations and only incorporates his valuation function in the allocation. Based on this, we define $\MMS_i$ equal to the minimum profit that any agent receives in this allocation, according to agent $\agent_i$'s valuation function. Obviously, in order to maximize $\MMS_i$, agent $\agent_i$ chooses an allocation that maximizes the minimum profit of the agents. We call an allocation fair (approximately fair), if every agent $\agent_i$ receives a set of items that is worth at least $\MMS_i$ (a fraction of $\MMS_i$) to him.

It is easy to see that $\MMS_i$ is the best possible guarantee that one can hope to obtain in this setting. If all agents have the same valuation function, then at least one of the agents receives a collection of items that are worth no more than $\MMS_i$ to him. A natural question that emerges here is whether a fair allocation with respect to $\MMS_i$'s is always possible? Although the experiments are in favor of this conjecture, \procacciafirst ~\cite{Procaccia:first} (EC'14) refuted this by an elegant and delicate counterexample. They show such a fair allocation is impossible in some cases, even when the number of agents is limited to 3. On the positive side however, they show an approximately fair allocation can be guaranteed. More precisely, they show that there always exists an allocation in which every agent's profit is at least $2/3\MMS_i$. Such an allocation is called a $2/3$-$\MMS$ allocation. \amanatidisapproximationful ~\cite{amanatidis2015approximation} later provided a proof for the existence of an $\MMS$ allocation for the case, when there are large enough items and the value of each agent for every items is drawn independently from a uniform distribution. A generalized form of this result was later proposed by Kurokawa \etal ~\cite{kurokawa2015can} for arbitrary distributions. Caragiannis \etal ~\cite{Carag} later proved that the maximum Nash welfare (MNW) solution, which selects an allocation that maximizes the product of utilities, for each agent guarantees a $2 / (1 + \sqrt{4n-3})$  fraction of her $\MMS$. In a recent work, Ghodsi \etal ~\cite{hadi34} provided a proof for existence of a $3/4$-$\MMS$ allocation. 

Although it is natural to assume the agents have equal entitlements on the items, in most real-world applications, agents have unequal entitlements on the goods. For instance, in various religions, cultures, and regulations, the distribution of the inherited wealth is often unequal. Furthermore, the division of mineral resources of a land or international waters between the neighboring countries is often made unequally based on the geographic, economic, and political status of the countries. 

For fairly allocating indivisible items to agents with different entitlements, two procedures are proposed in~\cite{brams1996fair}. The first one is based on \textit{Knaster's procedure of sealed bids}. In this method, we have an auction for selling each item. Therefore, for using it all the agents should have an adequate reserve of money which is the main issue of the procedure. The second procedure mentioned in~\cite{brams1996fair} is based on \textit{method of markers} developed by William F. Lucas which is spiritually similar to the moving knife procedure. In this method, first we line up the items, and then the agents place some markers for dividing the items. This method suffers from high dependency of its final allocation to the order of the items in the line. 

Agent duplication is another idea to deal with unequal entitlements. More precisely, when all of the entitlements are fractional numbers, we can duplicate each agent $a_i$ to some agents with similar valuation functions to $a_i$. The goal of this duplication is to reduce the problem to the case of equal entitlements. After the allocation, every agent $a_i$ owns all of the allocated items to her duplicated agents. For instance, assume that we have three agents with entitlements $1/2$, $2/5$, and $1/10$, respectively. In this case, we duplicate the first agent to five agents and the second agent to four agents each having an entitlement of $1/10$. This way, we can reduce our problem to the case of equal entitlements. Although agent duplication may be practical when the items are divisible, in the indivisible case, this method does not apply to the indivisible setting. For instance, if the number of the agents is higher than the number of available items, we cannot allocate anything to some agents. Another issue with this method is that it works only for fractional entitlements.

In this paper, we study fair allocation of indivisible items with different entitlements using a model which resolves the mentioned issues. Our fairness criterion mimics the general idea of Budish for defining maxmin shares. Similar to Budish's proposal, in order to define a maxmin share for an agent $\agent_i$, we ask the following question: how much benefit does agent $\agent_i$ expect to receive from a fair allocation, if we were to divide the goods \textit{only based on his valuation function}? If agent $\agent_i$ expects to receive a profit of $p$ from the allocation, then he should also recognize a minimum profit of $p\cdot \share_j/\share_i$ for any other agent $\agent_j$, so that his own profit per entitlement is a lower bound for all agents. Therefore, a fair answer to this question is the maximum value of $p$ for which there exists an allocation such that agent $\agent_i$'s profit-per-entitlement can be guaranteed to all other agents (according to his own valuation function). We define the maxmin shares of the agents based on this intuition.
 
Recall that we denote the number of agents with $n$ and the entitlement of every agent $\agent_i$ with $\share_i$. We assume the entitlements always add up to 1. For every agent $\agent_i$, we define the weighted maxmin share denote by $\WMMS_i$, to be the highest value of $p$ for which there exists an allocation of the goods to the agents in which every agent $\agent_j$ receives a profit of at least $p\cdot \share_j/\share_i$ based on agent $\agent_i$'s valuation function. Similarly, we call an allocation $\alpha$-$\WMMS$, if every agent $\agent_i$ obtains an $\alpha$ fraction of $\WMMS_i$ from his allocated goods. Notice that in case $\share_i = 1/n$ for all agents, this definition is identical to Budish's definition. Since our model is a generalization of the Budish's model, it is known that a fair allocation is not guaranteed to exist for every scenario. However, whether a $2/3$ approximation or in general a constant approximation $\WMMS$ allocation exists remains an open question.

Our main result is in contrast to that of \procacciafirst. We settle the above question by giving a $1/n$ hardness result for this problem. In other words, we show no algorithm can guarantee any allocation which is better than $1/n$-$\WMMS$ in general. We further complement this result by providing a simple algorithm that guarantees a $1/n$-$\WMMS$ allocation to all agents. As we show in Section \ref{noname}, this hardness is a direct consequence of unreasonably high valuation of agents with low entitlements for some items. Moreover, in Section \ref{hadi} we discuss that not only are such valuation functions unrealistic, but also an agent with such a valuation function has an incentive to misrepresent his valuations (Observation \ref{value}). Therefore, a natural limitation that one can add to the setting is to assume no item is worth more than $\WMMS_i$ for any agent $\agent_i$. We also study the problem in this mildly restricted setting and show in this case a $1/2$-$\WMMS$ guarantee can be delivered via a greedy algorithm. 

In contrast to our theoretical results, we show in practice a fair allocation is likely to exist by providing experimental results on real-world data. The source of our experiments is a publicly available collection of bids for eBay goods and services\footnote{\href{http://cims.nyu.edu/~munoz/data/}{http://cims.nyu.edu/~munoz/data/}}. Note that since those auctions are \textit{truthful}\footnote{An action is called truthful, if no bidder has any incentive to misrepresent his valuation}, it is the users' best interest to bid their actual valuations for the items and thus the market is transparent. More details about the experiments can be found in Section \ref{experimental}. We also support our claim by presenting theoretical analysis for the stochastic variants of the problem in which the valuation of every agent for a good is drawn from a given distribution. 

\subsection{Our Model}
Let $\agents$ be a set of $n$ agents, and $\items$ be a set of $m$ items. Each agent $\agent_i$ has an \textit{additive} valuation function $\valu_i$ for the items. In addition, every agent $\agent_i$ has an entitlement to the items, namely $\share_i$. The entitlements add up to $1$, i.e., $\sum \share_i = 1$. 

Since our model is a generalization of maxmin share, we begin with a formal definition of the maxmin shares for equal entitlements, proposed by Budish~\cite{Budish:first}. In this case, we assume all of the entitlements are equal to $1/n$. Let $\Pi(\items)$ be the set of $n$-partitionings of the items. Define the maxmin
share of agent $\agent_i$ ($\MMS_i$) of player $i$ as
\begin{equation}
\label{def1}
\MMS_i = \max_{\langle A_1, A_2, \ldots, A_n\rangle \in \Pi(\items)} \min_{j \in [n]} \valu_i(A_j)
.
\end{equation}
\color{black}
One can interpret the maxmin share of an agent as his outcome as a divider in a divide-and-choose procedure against adversaries~\cite{Budish:first}. Consider a situation that a cautious agent knows his own valuation on the items, but the valuations of other agents are unknown to him. If we ask the agent to run a divide-and-choose procedure, he tries to split the items in a way that the least valuable bundle is as attractive as possible. 

When the agents have different entitlements, the above interpretation is no longer valid. The problem is that the agents have different entitlements and this discrepancy must somehow be considered in the divide-and-choose procedure. Thus, we need an interpretation of the maxmin share that takes the entitlements into account. 

Let us get back to the case with the equal entitlements. Another way to interpret maxmin share is this: suppose that we ask agent $\agent_i$ to fairly distribute the items in $\items$ between $n$ agents of $\agents$, based on his own valuation function. In an ideal situation (e.g., if the goods are completely divisible), we expect $\agent_i$ to allocate a share with value  $\valu_i(\items) / n$ to every agent. However, since the goods are indivisible, some sort of unfairness is inevitable. For this case, we wish that $\agent_i$ does his best to retain fairness. $\MMS_i$ is in fact, a parameter that reveals how much fairness $\agent_i$ can guarantee, regarding his valuation function. 

Formally, to measure the fairness of an allocation by $\agent_i$, define a value $F^i_A$ for any allocation $A= \langle A_1,A_2,\ldots,A_n\rangle$ as $$F^i_A=\frac{\min_j \valu_i(A_j)}{\valu_i(\items) / n}. $$ In fact, we wish to make sure $\agent_i$ reports an allocation $A^*$ such that $F^i_{A^*}$ is as close to $1$ as possible. The maxmin share of $\agent_i$ is therefore defined as
\begin{equation}
\label{def2}
\MMS_i = F^i_{A^*} (\valu_i(\items) / n).
\end{equation}
It is easy to observe that Equations \eqref{def1} and \eqref{def2} are equivalent, since the fairest allocation in the absence of different entitlements is an allocation that maximizes value of the minimum bundle:

\begin{align*}
\MMS_i &= F^i_{A^*} (\valu_i(\items) / n)\\
	   &= \frac{\min_j \valu_i(A^*_j)}{\valu_i(\items) / n} (\valu_i(\items) / n) = \min_j \valu_i(A^*_j)
\end{align*}

Now, consider the case with different entitlements. Let $e_i$ be the entitlement of agent $\agent_i$. Similar to the second interpretation for $\MMS_i$, ask agent $\agent_i$ to fairly distribute the items between the agents, but this time, considers the entitlements. In an ideal situation (e.g., a completely divisible resource), we expect the allocation to be proportional to the entitlements, i.e. $\agent_i$ allocates a share to agent $\agent_j$ with value exactly $\valu_i(\items) e_j$ (note that when the entitlements are equal, this value equals to $\valu_i(\items)/n$ for every agent). But again, such an ideal situation is very rare to happen and thus we allow some unfairness. In the same way, define the fairness of an allocation $A = \langle A_1,A_2,\ldots,A_n\rangle$ as 

\begin{equation}
F^i_{A} = \min_j \frac{\valu_i(A_j)}{\valu_i(\items) e_j}
\end{equation}

Let $A^* = \langle A^*_1,A^*_2,\ldots,A^*_n\rangle$ be an allocation by $\agent_i$ that maximizes $F^i_{A^*}$. The weighted maxmin share of agent $\agent_i$ is defined in the same way as $\MMS_i$, that is:  

\begin{align*}
\WMMS_i &= F^i_{A^*} \valu_i(\items) e_i = e_i \min_j \frac{\valu_i(A^*_j)}{e_j} 
\end{align*}  
In summery, the value $\WMMS_i$ for every agent $\agent_i$ is defined as follows:
$$
\WMMS_i = \max_{\langle A_1, A_2, \ldots, A_n\rangle \in \Pi(\items)} \min_{j \in [n]}  \valu_i(A_j) \frac{e_i}{e_j}.
$$

For more intuition, consider the following example:
\begin{example}
Assume that we have two agents $\agent_1,\agent_2$ with $e_1 = 1/3$ and $e_2 = 2/3$. Furthermore, suppose that there are 5 items $b_1,b_2,b_3,b_4,b_5$ with the following valuations for $\agent_1$: $\valu_1(\{\ite_1\}) = \valu_1(\{\ite_2\}) = \valu_1(\{\ite_3\}) = 4,  \valu_1(\{\ite_4\}) = 3$ and $\valu_1(\{\ite_5\}) = 9$. For the allocation $A = \langle \{b_5\},\{b_1,b_2,b_3,b_4\} \rangle$, we have $F_A = \min (\frac{9}{24\cdot (1/3)},\frac{15}{24\cdot(2/3)})$ which means $F_A = 15/16$. Moreover, for allocation $A' = \langle \{b_1,b_2\},\{b_3,b_4,b_5\} \rangle$, we have $F_{A'} = \min (\frac{8}{24\cdot (1/3)},\frac{16}{24\cdot(2/3)})$ which means $F_{A'} = 1$. Thus, $A'$ is a fairer allocation than $A$. In addition, $A'$ is the fairest possible allocation and hence, $\WMMS_1 = 1 \cdot 24 \cdot 1/3 = 8$.
\label{ex1}
\end{example}
Example \ref{ex1} also gives an insight about why agent duplication (as introduced in the Introduction) is not a good idea. For this example, if we duplicate agent $\agent_2$, we have three agents with the same entitlements. But any partitioning of the items into three bundles, results in a bundle with value at most $7$ to $\agent_1$.

Finally, an allocation of the items in $\items$ to the agents in $\agents$ is said to be $\alpha-\WMMS$, if the total value of the share allocated to each agent $\agent_i$ is worth at least $\alpha \WMMS_i$ to him. 

\color{black}

%
\section{A Tight $1/n$ Bound on the Optimal Allocation}\label{noname}

In spite of the fact that there exists a $2/3$-$\WMMS$ guarantee when all the entitlements are equal, in our general setting surprisingly we provide a counterexample which proves that there is no guarantee better than $1/n$-$\WMMS$. We complement this result by showing that a $1/n$-$\WMMS$ always exists. Thus, these two theorems make a tight bound for the problem. 

The main property of our counterexample is a large gap between the value of items for different agents. We provide a counterexample according to this property in Theorem \ref{cex}.

\begin{theorem}
\label{cex}
There exists no guarantee better than $1/n$-$\WMMS$ when the entitlements to the items may differ.
\end{theorem}

\begin{proof}
We propose an example that admits no allocation better than $1/n$-$\WMMS$. To this end, consider an instance with $n$ agents and $2n-1$ items and let $\share_i = \epsilon$ for all $i <n$ and $\share_n = 1-(n-1)\epsilon$. The valuation functions of the first $n-1$ agents are the same. For every agent $\agent_i$ with $1 \leq i < n$, $\valu_i(\{\ite_j\})$ is as follows:
$$
\valu_i(\{\ite_j\}) = 
\begin{cases}
\epsilon & \mbox{if } j \leq n-1 \\
1-(n-1) \epsilon & \mbox{if } j = n \\
0 & \mbox{if } j > n.
\end{cases}
$$
Also, for agent $\agent_n$ we have:
$$
\valu_n(\{\ite_j\}) = 
\begin{cases}
\frac{1 - (n-1) \epsilon}{n} & \mbox{if } j \leq n \\
\epsilon & \mbox{if } j > n.
\end{cases}
$$
First, note that $\WMMS_i = \epsilon$ for the first $n-1$ agents and $\WMMS_n = 1 - (n-1) \epsilon$. For the first $n-1$ agents, the optimal partitioning is to allocate $\ite_i$ to $\agent_i$ for all $i \leq n$. Furthermore, the optimal partitioning for $\agent_n$ is to allocate items $\ite_{n+1},\ite_{n+2},\ldots,\ite_{2n-1}$ to the first $n-1$ agents and keep the first $n$ items for himself. In this case, $\valu_n(\{\ite_1,\ite_2,...,\ite_n\}) = 1- (n-1) \epsilon$, and $\valu_n(\{\ite_{n+i}\}) = \epsilon$ for $1 \le i <n$. Therefore, $\WMMS_n = 1- (n-1) \epsilon$.

On the other hand, in any allocation that guarantees a non-zero fraction of $\WMMS$ for every agent, at most one of the items $\ite_1,\ite_2,\ldots,\ite_n$ is allocated to $\agent_n$, since the rest of the items have value $0$ for the first $n-1$ agents. Therefore, the items allocated to $\agent_n$ are worth at most
$$
\frac{1 - (n-1) \epsilon}{n} + (n-1)\epsilon = 1/n +\epsilon(n-1 - \frac{n-1}{n})
\leq 1/n + n\epsilon.
$$
to him.
Thus, the best fraction of $\WMMS$ that can be guaranteed is 
\begin{equation}
\label{ce}
\frac{1/n + n \epsilon}{1 - (n-1)\epsilon}
\end{equation}

Equation \eqref{ce} can be made arbitrarily close to $1/n$, by choosing sufficiently small $\epsilon$. Thus, no allocation can guarantee an approximation better than $1/n$ for this example.

\end{proof}

\begin{algorithm}[t]
	\KwIn{$\agents$, $\items$, valuation functions $V_1, \ldots, V_n$, and entitlements $\share_1, \ldots, \share_n$ (without loss of generality, sorted in descending order).} 
	\KwOut{Allocation $A = A_1, \ldots, A_n$.} 
	
	\hrulefill
	
	\begin{algorithmic}[1]
		\FOR{$i$ from one to $|\items|$}
		\STATE {assign an unassigned item $\ite_j$ to $\agent_{i \bmod |\agents|}$ where $V_{i \bmod |\agents|}(\{\ite_j\})$ is maximum among unassigned items}
		\ENDFOR	
		\caption{$1/n$-$\WMMS$ allocation}
	\end{algorithmic}
	\label{greedy}
\end{algorithm}

Theorem \ref{cex} gives a $1/n$-$\WMMS$ upper-bound. In Theorem \ref{tight}, we show that the provided upper-bound is tight. Algorithm \ref{greedy} uses a simple greedy procedure, which is spiritually similar to an algorithm in~\cite{amanatidis2015approximation}, guaranteeing $1/n$-$\WMMS$ allocation as follow: sort the agents in descending order of entitlements. Starting from the first agent, ask every person to collect the most valuable item from the remaining set, one by one. Repeat the process until no more item is left. 

\begin{theorem}
\label{tight}
Algorithm \ref{greedy} guarantees a $1/n$-$\WMMS$ allocation.
\end{theorem}

\begin{proof}
If the number of items is smaller than the number of agents then the proof is trivial. Therefore, from this point on, we assume $m \geq n$.
Without loss of generality we assume agents are sorted in descending order of their entitlements, that is $\share_1 \geq \share_2 \ldots \geq \share_n.$ The goal is to prove that for each agent $\agent_i$ he receives at least $1/n$-$\WMMS_i$ using Algorithm \ref{greedy}. Suppose that the optimal allocation for $\agent_i$ is $B^{*} = \langle B^{*}_1, \ldots, B^{*}_n \rangle$. Without loss of generality suppose items are sorted according to their value for $\agent_i$ in descending order that is:
 $$V_i(\{b_1\}) \geq V_i(\{b_2\}) \geq \ldots \geq V_i(\{b_{|\items|}\}).$$
We call items $b_1$, $b_2$, \ldots, and $b_{i-1}$ \textit{heavy items} for $\agent_i$. Let $H$ be the set of heavy items for $\agent_i$. Since the entitlements of agents are sorted from $\agent_1$ to $\agent_n$, for agents $\agent_j$ and $\agent_{j'}$ when $\share_j > \share_{j'}$ we have: 
\begin{equation}
\label{napprox1}
V_i(B^{*}_{j}) \geq V_i(B^{*}_{j'}).
\end{equation}
Now, the goal is to prove $V_i(B^{*}_i) \leq V_i(\items \setminus H)$ and use it to the guarantee of the algorithm. Since $B^{*}_i \subseteq \items$, if $B^{*}_i \cap H = \emptyset$, clearly $V_i(B^{*}_i) \leq V_i(\items \setminus H)$ holds. In case $B^{*}_i \cap H \neq \emptyset$, consider agent $\agent_{k'}$ where $1 \le k' <i$ and $B^{*}_{k'} \cap H = \emptyset$.  According to Inequality (\ref{napprox1}), $\agent_{k'}$ has a greater entitlement than $\agent_i$. Therefore, we have: $V_i(B^{*}_{k'}) \geq V_i(B^{*}_i)$ which yields that $B^{*}_{k'}$ is worth at least $V_i(\{b_{k}\})$ for $\agent_i$ where all the items of $B^{*}_{k'}$ are in $\items - H$. Therefore, we can imply that $V_i(B^{*}_i) \leq V_i(\items \setminus H).$
The items are sorted in the descending order of their value to $\agent_i$ from $b_1$ to $b_{|\items|}$ which implies $$V_i(\{b_i\}) + V_i(\{b_{n + i}\}) + V_i(\{b_{2n + i}\}) + \ldots \geq V_i(\items \setminus H)/n.$$ 
The $l$-th assigned item to $\agent_i$ by the algorithm is not worth less than $V_i(\{b_{(l-1)n + i}\})$. Hence, the assigned items to $\agent_i$ is worth at least $V_i(\{b_i\}) + V_i(\{b_{n + i}\}) + V_i(\{b_{2n + i}\}) + \ldots$ for him which is not less than $V_i(\items \setminus H)/n \geq V_i(B^{*}_i)/n \geq 1/n$-$\WMMS_i$.

\end{proof}

\section{Allocation for the Restricted Case}\label{hadi}

In Section \ref{noname}, we gave a tight $1/n$-$\WMMS$ guarantee for the fair allocation problem with unequal entitlements . In this section, we consider a reasonable restriction of the problem which gives a $1/2$-$\WMMS$ guarantee. In this restricted setting the value of each item $\ite_j$ to each agent $\agent_i$ is no more than $\WMMS_i$. Observation \ref{value} shows that how this assumption can be invaluable. 

\begin{observation}
\label{value}
If $V_i(\ite_j) \geq \WMMS_i$, it is in the best interest of $\agent_i$ to report his valuation for $\ite_j$ equal to infinity. Because his $\WMMS$ may increase in this way, and he achieves more items after the allocation of the algorithm, because his $\WMMS$ may increase in this way, and he will be satisfied even if he receives only this item.
\end{observation}

In this section, we provide an algorithm with $1/2$-$\WMMS$ guarantee for the restricted case of the problem. For a case that the entitlements are equal, an algorithm, namely \emph{bag filling} guarantees $1/2$-$\WMMS$ for all the agents. In the bag filling algorithm, we start with an empty bag. In each step, we add a remaining item to the bag. After each addition, if the total value of items in the bag becomes more than $1/2$-$\WMMS_i$ for an unsatisfied agent $\agent_i$, we allocate all the items in the bag to him and repeat the procedure with an empty bag. After running this simple procedure, each agent $\agent_i$ receives at least $1/2$-$\WMMS_i$.

Our algorithm allocates the items to agents in a more clever way. We allocate an item to an agent in each step of the algorithm until each agent $\agent_i$ receives at least $1/2$-$\WMMS_i$ by the allocation. To this end, in each step, first the algorithm for each agent $\agent_i$ creates a candidate set of items where $\ite_j$ is in the candidate set of $\agent_i$ if $V_i(\{\ite_j\})/V_i(\items)F^i_{A^{*}}$ be the maximum number among all of the unsatisfied agents. Then, the algorithm chooses unsatisfied agent $\agent_i$ and item $\ite_j$ in its candidate set which maximize $V_i(\{\ite_j\})/V_i(\items)F^i_{A^{*}}$ among all of the unsatisfied agents and items in their candidate sets, and assigns this item to the agent.

\begin{algorithm} [t]
    \KwIn{$\agents$, $\items$, valuation functions $V_1, \ldots, V_n$, and entitlements $\share_1, \ldots, \share_n$ (without loss of generality, sorted in descending order).} 
    \KwOut{Allocation $A = A_1, \ldots, A_n$.} 

\hrulefill

\begin{algorithmic}[1]

    \WHILE{$\items \neq \emptyset$ and $\exists A_i$ where $V_i(A_i) < 1/2$-$\WMMS_i$}
        \STATE {define candidate set $C_i = \emptyset$ for each agent $\agent_i$.}    
        \FOR{$\ite_j \in \items$}
            \STATE {add $\ite_j$ to $C_i$ where $V_i(\{b_j\})/V_i(\items)F^i_{A^{*}}$ is maximum among all of the unsatisfied agents.}
        \ENDFOR
        \STATE { Choose unsatisfied agent $\agent_i$ and item $\ite_j$ in its candidate set which maximize $V_i(\{\ite_j\})/V_i(\items)F^i_{A^{*}}$ among all of the unsatisfied agents and items in their candidate sets.}
        \STATE { Assign item $\ite_j$ to agent $\agent_i$ }
    \ENDWHILE
 \caption{$1/2$-$\WMMS$ allocation for the restricted case}
   \label{restricted}
 \end{algorithmic}
 
\end{algorithm}

Before proving Algorithm \ref{restricted} guarantees a $1/2$-$\WMMS$ allocation, we prove an auxiliary lemma which argues that using Algorithm \ref{restricted} no agent receives more than his $\WMMS$.

\begin{lemma}
\label{restricted_lemma}
Algorithm \ref{restricted} does not allocate more than $\WMMS_i$ for any agent $\agent_i$.
\end{lemma}

\begin{proof}
The algorithm does not assign anymore items to a satisfied agent. Hence, any satisfied agent $\agent_i$ has less than $1/2$-$\WMMS_i$ value of items in $A_i$ before he receives the last item. For the sake of contradiction, suppose that the algorithm allocates more than $\WMMS_i$ to $\agent_i$. Without loss of generality, suppose that the last item allocated to $\agent_i$ is $\ite_j$. Since before allocating $\ite_j$ to $\agent_i$ he had less than $1/2$-$\WMMS_i$ value of items in $A_i$, the value of $\ite_j$ to $\agent_i$ is more than $1/2$-$\WMMS_i$. Since $V_i(b_j) \leq \WMMS_i$, before allocating $\ite_j$ to $\agent_i$ we have $A_i \neq \emptyset$. Since $V_i(b_j)$ is more than the value of all the other allocated items to $\agent_i$, $\ite_j$ was not in the candidate set of $\agent_i$ when we were assigning the other items to $\agent_i$. Hence, $\ite_j$ was in the candidate set of an other agent $\agent_{i'}$. Therefore, $V_{i'}(\{\ite_j\})/V_{i'}(\items)F^{i'}_{A^{*}} \ge V_i(\{\ite_j\})/V_i(\items)F^i_{A^{*}}$. Since this value is greater than the values of all other items in $A_i$, the algorithm first allocate $\ite_j$ to unsatisfied agent $\agent_{i'}$.
\end{proof}

Now, using Lemma \ref{restricted_lemma}, we prove the approximation guarantee of the algorithm.

\begin{theorem}
\label{restrictedalgorithm}
Algorithm \ref{restricted} ensures a $1/2$-$\WMMS$ guarantee when for each agent $\agent_i$ and item $\ite_j$, $V_i(\ite_j) \leq \WMMS_i.$
\end{theorem}

\begin{proof}
It is clear that if the algorithm satisfies all the agents, it ensures the approximation guarantee. Now, for the sake of contradiction assume that there exists an unsatisfied agent $\agent_i$ at the end of the algorithm. For each item $\ite_j$ we define$$v'_{b_j} = V_{i'}(\ite_j)/V_{i'}(\items)F^{i'}_{A^{*}}$$ where $\agent_{i'}$ is the recipient of $\ite_j$ in the allocation.
Since $v'_{b_j}$ is maximal according to the algorithm, and $\agent_i$ is not a satisfied agent, we have: 
\begin{equation}
\label{uneq1}
V_i(\ite_j)/V_i(\items)F^i_{A^{*}} \leq V_{i'}(\ite_j)/V_{i'}(\items)F^{i'}_{A^{*}}.
\end{equation}
Lemma \ref{restricted_lemma} implies that $\sum_{b_j}{v'_{b_j}} \leq 1$, and since we have at least one unsatisfied agent we can write:
\begin{equation}
\label{uneq2}
\sum_{b_j}{v'_{b_j}} < 1.
\end{equation}
Inequality (\ref{uneq1}) along with Inequality (\ref{uneq2}) implies
\begin{equation}
\label{uneq3}
\sum_{b_j}{V_i(b_j)/V_i(\items)F^i_{A^{*}}} < 1.
\end{equation}
Finally Inequality \eqref{uneq3} yields that $\sum_{b_j}{V_i(b_j)} < V_i(\items)F^i_{A^{*}}$ which is a contradiction.
\end{proof}

\section{Empirical Results}\label{experimental}
As we discussed in Section \ref{noname}, in extreme cases, making a $\WMMS$ allocation or even an approximately $\WMMS$ allocation is theoretically impossible. However, our counter-example is extremely delicate and thus very unlikely to happen in real-world. Here, we show in practice fair allocations w.h.p exist, especially when the number of items is large.

\begin{figure}[t]
	\begin{center}
	\includegraphics[width=12cm]{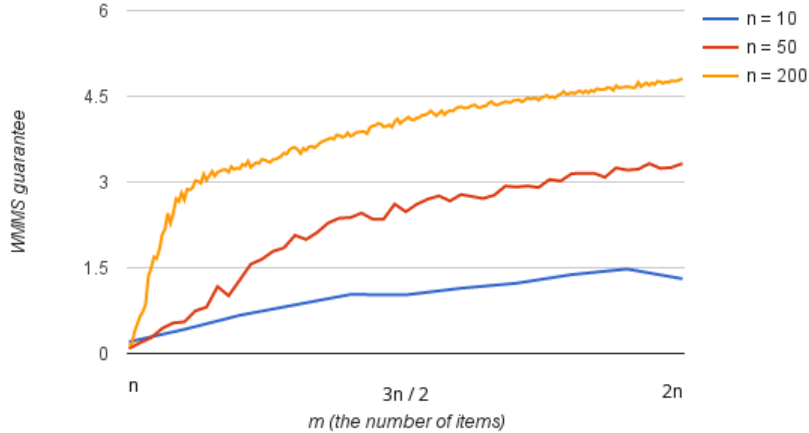}\vspace{-0.5cm}	
	\end{center}		
	\caption{The vertical line denotes the ratio of the valuation of the allocated set to the maxmin guarantee of the agents. The horizontal line shows the number of items varying from $n$ to $2n$. Blue, red, and yellow poly lines illustrate the performance of our algorithm for $n=10$, $n=50$, and $n=200$ respectively.}
	\label{gudgud1}
\end{figure}

\begin{figure}[t]
	\begin{center}
		\includegraphics[width=12cm]{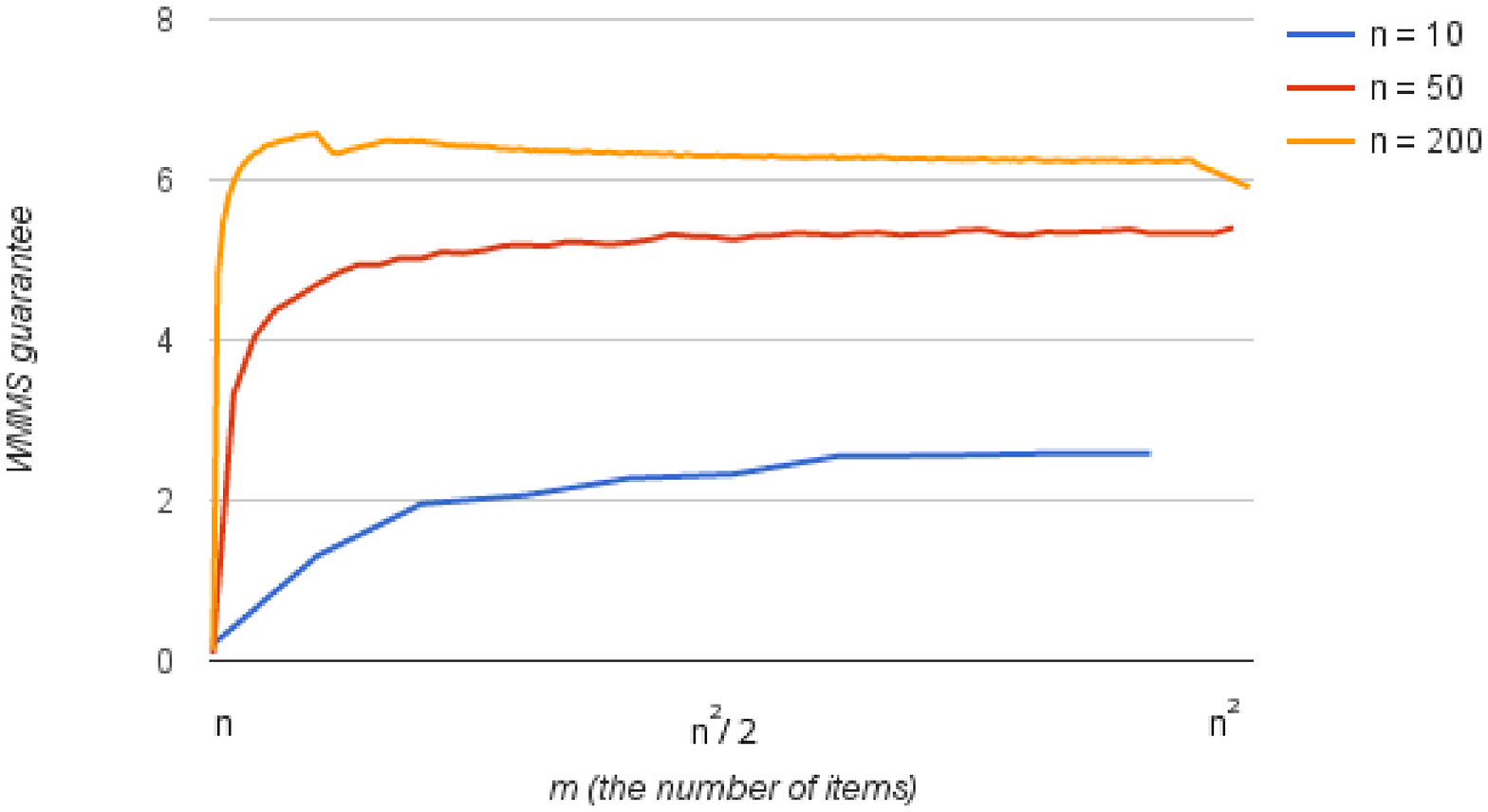}\vspace{-0.5cm}
	\end{center}
	\caption{The vertical line denotes the ratio of the valuation of the allocated set to the maxmin guarantee of the agents. The horizontal line shows the number of items varying from $n$ to $n^2$. Blue, red, and yellow polylines illustrate the performance of our algorithm for $n=10$, $n=50$, and $n=200$ respectively.}
	\label{gudgud2}
\end{figure}

%
We draw the valuation of the agents for the goods based on a collection of bids for eBay items publicly available at \href{http://cims.nyu.edu/~munoz/data/}{http://cims.nyu.edu/~munoz/data/}. More precisely, for $m$ items, we randomly choose $m$ different categories of goods from the dataset. Moreover, for every agent $\agent_i$ and item $\ite_j$, we set $\valu_i(\{\ite_j\})$ to a submitted bid for the corresponding category of item $\ite_j$ chosen uniformly at random. The bids vary from 0.01 to 113.63 and their mean is 6.57901. Moreover, the expected variance of the bids in every category is 200.513.

For an instance of the problem with $n$ agents and $m$ items, we run the experiments with 1000 different vector of entitlements drawn from the uniform distribution (and scaled up to satisfy $\sum \share_i = 1$). For every $n$ and $m$, we take the minimum $\WMMS$ guarantee obtained all 1000 runs, and show it in Figures \ref{gudgud1} and \ref{gudgud2}. We used heuristic algorithms to compute the maxmin shares and maxmin guarantees. Thus, our results are only \textit{lower bounds} to the actual $\WMMS$ guarantees. Nonetheless, the optimal guarantees are very close to the estimated ones.

Figures \ref{gudgud1} and \ref{gudgud2} illustrate the result of the runs for $n=10$, $n=50$, and $n=200$ respectively. Figure \ref{gudgud1} only depicts the $\WMMS$ guarantees for $m \in [n,2n]$ whereas in Figure \ref{gudgud2} the number of items varies from $n$ to $n^2$.

As shown in Figures \ref{gudgud1} and \ref{gudgud2}, the approximation guarantee improves as we increase the number of items. Moreover, unless $m$ is very close to $n$, a $\WMMS$ allocation exists in our experiments (notice that the guarantee is above 1 when $m$ is considerably larger that $n$).  
\section{Stochastic Setting}\label{random}
In Section \ref{noname} we presented a counterexample to show that no allocation better than $1/n$-$\WMMS$ can be guaranteed. However, the construction described in the counterexample is very unlikely to happen in the real settings.  Here, we show that $\WMMS$ allocation exists with high probability when a small randomness is allowed in the setting.

Considering stochastic settings is common in the fair allocation problems since many real-world instances can be modeled with random distributions~\cite{Budish:first,kurokawa2015can,amanatidis2015approximation,dickerson2014computational}. The general probabilistic model used in previous works is as follows: every agent $\agent_i$ has a probability distribution ${\cal D}_i$ over $[0,1]$ and for every item $\ite_j$, the value for $\valu_{i}(\{\ite_j\})$ is randomly sampled from ${\cal D}_i$. In~\cite{amanatidis2015approximation}, the existence of an $\MMS$ allocation is proved for the special case of ${\cal D}_i = U(0,1)$, where $U(0,1)$ is the standard uniform distribution with minimum $0$ and maximum $1$.  Kurokawa, Procaccia, and Wang~\cite{kurokawa2015can} considered the problem for arbitrary random distribution ${\cal D}_i$ with the condition that ${\mathbb V}[{\cal D}_i] \geq c$ for a positive constant $c$. A considerable part of the proof for the existence of an $\MMS$ allocation in~\cite{kurokawa2015can} is referred to~\cite{dickerson2014computational}, where the authors proved the existence of an envy-free allocation in the stochastic settings with arbitrary random distributions. 

In this section, we consider two different probabilistic models. Our first model is the same as~\cite{kurokawa2015can}, with the exception that we omit the restriction ${\mathbb V}[{\cal D}_i] \geq c$. We name this model as \emph{Stochastic Agents} model. In the second model, every item $\ite_i$ has a probability distribution ${\cal D}_i$ and for every agent $\agent_j$, value of $\valu_{j}(\{\ite_i\})$ is randomly drawn from ${\cal D}_i$. We choose the name \emph{Stochastic Items} for the second model. We believe that \textit{Stochastic Items} model is more realistic since the first model does not make any distinguish between the items. None of the previous works mentioned above considered this model. 

We leverage \emph{Hoeffding} inequality to prove the existence of $\WMMS$ allocation. Theorem \ref{hoeffding} states the general form of this inequality~\cite{hoeffding1963probability}. 
\begin{theorem}[General Form of  Hoeffding (1963)]
\label{hoeffding}
Let $X_1,X_2, \ldots, X_n$ be random variables bounded by the interval $[0,1]: 0 \leq X_i \leq 1$. We define the empirical mean of these variables by 
$\bar{X} = \frac{1}{n} (X_1+X_2+ \ldots + X_n)$.
Then, the following inequality holds:
\begin{equation}
\label{Hoeffding}
{\mathbb P}(|\bar{X} - E(\bar{X})| \geq t) \leq 2e^{-2nt^2}
\end{equation}
\end{theorem}
Regarding Theorem \ref{hoeffding}, let $X = n \bar{X} = \sum_i X_i$ and let $\mu = n \times E(\bar{X})$. By Inequality \eqref{Hoeffding}, we have:
\begin{equation}
\label{eq1}
{\mathbb P}(|X - \mu| \geq nt) \leq 2e^{-2nt^2}
\end{equation}
By setting $nt = \delta \mu$, we rewrite Equation \eqref{eq1} as: 
\begin{equation}
\label{eq2}
{\mathbb P}(|X - \mu| \geq \delta \mu) \leq 2e^{\frac{-2(\delta \mu)^2}{n}}
\end{equation}

\subsection{Model I: Stochastic Agents}
As mentioned before, in the first model we assume that every agent has a probability distribution ${\cal D}_i$ and for every item $\ite_j$, the value of $\valu_{i}(\{\ite_j\})$ is randomly sampled from ${\cal D}_i$. Furthermore, we suppose $\mu_i = {\mathbb E}({\cal D}_i)$. 
Throughout this section, we assume that $m \geq n$. This is w.l.o.g
\color{black}
, because for the case $m < n$, $\WMMS_i$ for every agent $\agent_i$ equals to zero. 
For the Stochastic Agents model, we state Theorem \ref{model1}.

\begin{theorem}
\label{model1}
Consider an instance of the fair allocation problem with unequal entitlements, such that the value of every item for every agent $\agent_i$ is randomly drawn from distribution ${\cal D}_i$. Furthermore, let $s = \min_i s_i$ and $\mu = \min_i \mu_i$. Then, for every $0<\epsilon<1$, there exists a value $m' = m'(\frac{1}{\share},\frac{1}{\mu},\frac{1}{\epsilon})$ (which means $m'$ is a function of $\frac{1}{\share},\frac{1}{\mu},\frac{1}{\epsilon}$) such that if $m \geq m'$, then almost surely a $(1-\epsilon)$-$\WMMS$ allocation exists.
\end{theorem}

In the rest of this section, we prove Theorem \ref{model1}. Consider the algorithm that allocates $t_i = \lfloor m \share_i \rfloor$ items to every agent $a_i$. We know that the value of item $\ite_i$ for agent $\agent_j$ is randomly sampled from ${\cal D}_i$. From the point of view of the algorithm, it trivially does not matter whether the value of items are sampled after the allocation or before the allocation. Thus, we can suppose that the value of every item is sampled after the allocation of items. 

For now, we know that $t_i = \lfloor m \share_i \rfloor$ number of items are assigned to $\agent_i$. Argue that for every $\epsilon > 0$, there exists a value $m'$, such that for every $m \geq m'$, $t_i \geq m \share_i (1-\epsilon)$. In Lemma \ref{b1} we bound the value of $m'$ in terms of $\share_i$ and $\epsilon$.

\begin{lemma}
\label{b1}
For every $m > \frac{1}{\epsilon \share_i}$, $\lfloor m \share_i \rfloor \geq m\share_i(1-\epsilon)$. 
\end{lemma}
\begin{proof}
We know $
\lfloor m \share_i \rfloor \geq m \share_i - 1 = m \share_i (1- \frac{1}{m \share_i})
$. If $m > \frac{1}{\epsilon \share_i}$, then 
$
\lfloor m \share_i \rfloor  \geq m\share_i(1-\frac{1}{\frac{1}{\epsilon e_i}e_i}) = m \share_i(1-\epsilon)
$.
This, completes the proof.
\end{proof}

For the rest of the proof, suppose that $m > \frac{1}{\epsilon \share_i}$. 
Let $X_i$ be the variable indicating total value of items allocated to $\agent_i$. Note that $\mathbb E(X_i) = t_i \mu_i$. Regarding Equation \eqref{eq2}, we have:
$$
{\mathbb P}(|X_i - t_i \mu_i| \geq \delta t_i \mu_i) \leq 2e^{\frac{-2 (\delta t_i \mu_i )^2}{t_i}}
$$
We want to choose $\delta$ such that ${\mathbb P}(|X_i - t_i \mu_i| \geq \delta t_i \mu_i) \leq \frac{1}{2mn}$. We have: 
$$
2e^{\frac{-2 (\delta t_i \mu_i )^2}{t_i}} \leq \frac{1}{2mn} \Rightarrow -2 (\delta t_i \mu_i )^2 \leq t_i \ln \frac{1}{4mn} \Rightarrow  \delta \geq \sqrt{ \frac{t_i \ln 4mn}{2  (t_i \mu_i )^2} }
$$

Regarding the facts that $t_i \geq m\share_i(1-\epsilon)$ and $m \geq n$, we have:
$$
\sqrt{ \frac{t_i \ln 4mn}{2  (t_i \mu_i )^2} } \leq \sqrt{\frac{\ln2 + \ln m}{ \share_i m {\mu_i} ^2(1-\epsilon)} } 
$$
Therefore, it's enough to choose $\delta$ such that
\begin{equation}
\label{delta1}
\delta \geq \sqrt{\frac{\ln 2 + \ln m}{ \share_i m {\mu_i} ^2(1-\epsilon)} }. 
\end{equation}

Now, let $t'_i $ be the number of items that are not assigned to $\agent_i$. Since $t_i + t'_i = m$, regarding the fact that $t_i \geq m\share_i(1-\epsilon)$, we have $t'_i \leq m - m\share_i(1-\epsilon) $, which means $t'_i \leq m + m\share_i(\epsilon-1)$. On the other hand, $t'_i \geq m(1-\share_i)$. Also, let $X'_i$ be the variable indicating total value of the items that are not allocated to $\agent_i$. By the same deduction as $t_i$ for $t'_i$ we have:
$$
 {\mathbb P}(|X'_i - t'_i \mu_i| \geq \delta' t'_i \mu_i) \leq 2e^{\frac{-2 (\delta' t'_i \mu_i )^2}{t'_i}}.
$$
Let $\delta'$ be the value that ${\mathbb P}(|X_i - t'_i \mu_i| \geq \delta t_i' \mu_i)< \frac{1}{2mn}$. We have: 
$$
2e^{\frac{-2 (\delta' t'_i \mu_i )^2}{t'_i}} \leq \frac{1}{2mn} \Rightarrow (\delta' t'_i \mu_i )^2 \geq \frac{t'_i\ln{4mn}}{2}
$$
\begin{equation}
\label{eqprin}
\delta' \geq \sqrt{ \frac{ \ln 4mn}{2   t'_i \mu_i ^2} } 
\end{equation}

Thus, it's enough to choose $\delta$ in a way that Inequality \eqref{eqprin} holds. Regarding the facts that $m>n$ and $t'_i \geq m(1-\share_i)$, 
$$
\sqrt{ \frac{ \ln 4mn}{2   t'_i \mu_i ^2} }  \leq \sqrt{ \frac{\ln 2 + \ln m}{ m(1-\share_i) {\mu_i} ^2 } }.
$$
Therefore, it's enough to choose $\delta'$ in a way that 
\begin{equation}
\label{delta2}
\delta' \geq  \sqrt{ \frac{\ln 2 + \ln m}{ m(1-\share_i) {\mu_i} ^2 } } 
\end{equation}

Now, suppose that both Inequalities \eqref{delta1} and \eqref{delta2} are held. Considering $S_i$ as the set of items assigned to $\agent_i$, with the probability of at least $1-(\frac{1}{2mn} + \frac{1}{2mn}) = 1 - \frac{1}{mn}$ we have:
$$
\frac{\valu_i(S_i)}{\valu_i(\items)} = 
\frac{(1-\delta)t_i \mu_i}{(1+\delta)t_i\mu_i + (1+\delta')t'_i\mu_i}
$$
It is easy to show that, there always exist an $m'_i$ such that for all $m \ge m'_i$ both Inequalities \eqref{delta1} and \eqref{delta2} hold for $\delta = \epsilon$ and $\delta' = \epsilon$. Regarding this, we have:
\begin{align*}
\frac{\valu_i(S_i)}{\valu_i(\items)} &\geq \frac{(1-\epsilon)t_i \mu_i}{(1+\epsilon)t_i\mu_i + (1+\epsilon)t'_i\mu_i}\\
&=\frac{(1-\epsilon)t_i}{(1+\epsilon)m}\\
&\geq \frac{(1-\epsilon)m\share_i(1-\epsilon)}{(1+\epsilon)m}\\
& = \frac{(1-\epsilon)^2}{1+\epsilon} \share_i = \frac{(1+\epsilon)^2-4\epsilon}{1+\epsilon} \share_i\\
& = (1 +\epsilon) \share_i - \frac{4\epsilon}{(1+\epsilon)}\share_i \\
& \geq (1 +\epsilon) \share_i -4\epsilon\share_i = (1-3\epsilon) \share_i
\end{align*}
Therefore, with the probability at least $1 - \frac{1}{mn}$ we have:
\begin{equation}
\label{lastin}
\forall_{\agent_i \in \agents} \qquad \frac{\valu_i(S_i)}{\valu_i(\items)} \geq (1-3\epsilon) \share_i
\end{equation}
Now, suppose that the Inequality \ref{lastin} holds for every agent $\agent_i$, with probability at least $1 - \frac{1}{mn}$. Considering all the agents, with the probability at least $(1-\frac{1}{mn})^n \geq 1 - \frac{n}{nm} = 1-\frac{1}{m}$, value of $\frac{\valu_i(S_i)}{\valu_i(\items)}$ for every agent $a_i$ is at least $\share_i(1-3\epsilon)$. Regarding the fact that $\WMMS_i \leq \share_i \valu_i(\items)$, we have 
$$
\forall_{\agent_i \in \agents} \qquad \valu_i(S_i) \geq \WMMS_i (1-3\epsilon)
$$
This completes the proof.
\bibliographystyle{abbrv}

\subsection{Model II: Stochastic Items}
As mentioned, in Stochastic Items model, every item $\ite_i$ has a probability distribution ${\cal D}_i$ and the value of every agent $\agent_j$ for item $\ite_i$ is randomly chosen from ${\cal D}_i$. For this model, we prove Theorem \ref{model2}. The theorem states that for large enough $m$, almost surely a $(1-\epsilon)$-$\WMMS$ allocation exists. 
\begin{theorem}
\label{model2}
Suppose that for every agent $\agent_i$, ${\mathbb E}({\cal D}_i) > c$ for a non-negative constant $c$. Then for all $ 0<\epsilon<1$, there exists $m'=m'(c,\epsilon,n,\share_1,...,\share_n)$ such that if $m \ge m'$, then, almost surely, $(1-\epsilon)$-$\WMMS$ allocation exists. 
\end{theorem}
In the rest of the section, we prove Theorem \ref{model2}. First, in Lemma \ref{lem1} we prove the existence of an allocation that assigns to every agent $\agent_i$, a set of items with value at least $\WMMS_i - M_i$, where $M_i = \max_j(\valu_i(\{\ite_j\}))$. The idea to prove this fact is inspired by~\cite{bezakova2005allocating}. Argue that we can formulate the allocation problem with unequal entitlements as the following Integer program:
\begin{equation}
\label{LP1}
\begin{array}{ll@{}ll}
& & \displaystyle\sum_{\ite_j \in \items} \valu_i(\{\ite_j\}) \cdot f_{i,j} \geq \valu(\items) \cdot \share_i & \displaystyle\forall_{\agent_i \in \agents}\\

&    &\displaystyle\sum_{\agent_i \in \agents} f_{i,j} = 1 &\displaystyle\forall_{\ite_j \in \items}\\
                 &                                                &f_{i,j} \in \{0,1\}
\end{array}
\end{equation}
In IP\eqref{LP1}, variable $f_{i,j}$ determines whether $\ite_j$ is assigned to agent $\agent_i$ or not. Considering the fact that $\valu_i(\items)\cdot \share_i$ is a trivial upper bound on $\WMMS_i$, any solution to IP\ref{LP1} is a feasible solution to the assignment problem with unequal entitlements. 
 By relaxing the second and the third condition, we can convert IP\ref{LP1} to LP\ref{LP2}:
\begin{equation}
\label{LP2}
\begin{array}{ll@{}ll}
& & \displaystyle\sum_{\ite_j \in \items} \valu_i(\{\ite_j\}) \cdot f_{i,j} \geq \valu(\items) \cdot \share_i & \displaystyle\forall_{\agent_i \in \agents}\\

&    &\displaystyle\sum_{\agent_i \in \agents} f_{i,j} \leq 1 &\displaystyle\forall_{\ite_j \in \items}\\
                 &                                                & f_{i,j} \ge 0
\end{array}
\end{equation}
For every feasible solution $\cal A$ to LP\ref{LP2}, we construct the bipartite graph $G_{\cal A} \langle I,J,E \rangle$ where $I=\{1,2,...,n\}$ and $J=\{1,2,..,m\}$ correspond to the set of players and items, respectively. An edge $i,j$ is included if $f_{i,j} > 0$. Then by the same way used by~\cite{bezakova2005allocating}, we will prove the following theorem.
\begin{lemma}
\label{special}
There exists a solution $\cal A'$ to LP\ref{LP2}, such that $G_{\cal A}$ is a pseudoforest. (each component of the graph is either a tree or a tree with an extra edge)
\end{lemma}
\begin{proof}
We have $mn+m+n$ inequalities defining the polytope of feasible solutions of LP\ref{LP2}. We have $mn$ variables $f_{i,j}$, therefore every solution which is located in the corner of polytope satisfies at least $mn$ inequalities as equalities and there will be at most $m+n$ non-zero variables in these solutions. By the same method used by~\cite{bezakova2005allocating}, it is clear to show that if $\cal A'$ is corresponding solution to a corner of polytope, then $G_{\cal A'}$ is a pseudoforest.
\end{proof}

We call the solution with the property defined in Lemma \ref{special} as constrained solution. In~\cite{bezakova2005allocating} it is shown that every constrained solution for LP\ref{LP2}, can be converted to a solution for IP\ref{LP1}, such that every agent $a_i$ loses at most one item $b_j$ where $f_{i,j}>0$.

\begin{lemma}
\label{lem1}
There exists an allocation in which every agent $\agent_i$ gets at least $\WMMS_i - \max_j V_i(\{b_j\})$.
\end{lemma}
\begin{proof}
The polytope of of feasible solutions of LP\ref{LP2} is non-empty, because it has at least one solution which is, $f_{i,j}=\dfrac{1}{\share_i}$ for every agent $a_i$ and item $b_j$. Therefore there exists a constrained solution for LP\ref{LP2}, and using the same method in \cite{bezakova2005allocating} this solution can be converted to a solution for IP\ref{LP1}, such that every agent loses at most one item.
\end{proof}

Proofs of Lemmas \ref{special} and \ref{lem1} are omitted and included in the full version.\\

Now we show that there exists $m'_i= m'_i(c,\epsilon,n,\share_1,...,\share_n)$, such that if $m \ge m'_i$, then, $\WMMS_i \ge \dfrac{1}{\epsilon}$ with the probability at least $1-\dfrac{1}{m n}$. Therefore, whenever $m \ge \max(m'_1,m'_2,...,m'_n)$, with the probability at least $(1-\dfrac{1}{mn})^n \geq 1 - \dfrac{n}{nm} = 1-\dfrac{1}{m}$, for every agent $a_i$, $\WMMS_i \ge \dfrac{1}{\epsilon}$. According to Lemma \ref{lem1}, there is an allocation in which every agent gets at least
\begin{align*}
\WMMS_i - \max_j V_j(\{b_j\}) &\ge \WMMS_i -1 \\ 
&\ge \WMMS_i ( 1- \dfrac{1}{\WMMS_i})\\ &\ge \WMMS_i(1- \epsilon) 
\end{align*}
Therefore, a $(1-\epsilon)-\WMMS$ allocation is guaranteed to exist with the probability at least $1-\dfrac{1}{m}$.
\begin{lemma}
Suppose that for every agent $\agent_j$, ${\mathbb E}({\cal D}_j) > c$ for a non-negative constant $c$. Then for all $ 0<\epsilon<1$, there exists $m'_i=m'_i(c,\epsilon,n,\share_1,...,\share_n)$ such that if $m \ge m'_i$, $\WMMS_i \ge \dfrac{1}{\epsilon}$ with the probability at least $1- \dfrac{1}{mn}$.
\end{lemma}
\begin{proof}
If there exists a partition of $\items$, $\pi = B_1, B_2, ..., B_n$ in which $\valu_i(B_j) \ge \dfrac{1}{\epsilon \share_i}$ for every agent $\agent_j$, then:
\begin{equation}
\WMMS_i \ge \share_i \cdot \min(\dfrac{1}{\epsilon \share_i \share_1},...,\dfrac{1}{\epsilon \share_i \share_n}) \ge \min(\dfrac{1}{\epsilon \share_1},...,\dfrac{1}{\epsilon \share_n}) \ge \dfrac{1}{\epsilon}
\end{equation}
Suppose that $m=\dfrac{\alpha}{\epsilon}$, then for every item $b_k$ and every agent $a_j$, we assign this item to this agent with the probability $\share_j$. Let $X_{j,k}$ be a random variable that takes the value $V_i(\{b_k\})$ with the probability $\share_j$ and 0 otherwise. We have ${\mathbb E}[X_{j,k}] > \dfrac{c}{\share_j}$. Let $X_j= \sum_{k=1}^{m} X_{j,k}$. By setting $nt = \gamma$, we rewrite Equation \eqref{eq1} as: 
\begin{equation}
\label{eq3}
{\mathbb P}(|X - \mu| \geq \gamma) \leq 2e^{\frac{-2 \gamma^2}{n}}
\end{equation}
${\mathbb E}(X_j)$ will be at least $\dfrac{\alpha c}{\epsilon \share_j}$. Regarding Equation \eqref{eq3}, we have:
\begin{align*}
&P(|X_j- {\mathbb E}(X_j)| \ge \dfrac{\alpha c}{\epsilon \share_j} -\dfrac{1}{\epsilon \share_i}) \le 2 e^{\dfrac{-2 \epsilon}{\alpha} (\dfrac{\alpha c}{\epsilon \share_j} -\dfrac{1}{\epsilon \share_i})^2} \\ &\Rightarrow P(|X_j- {\mathbb E}(X_j)| \ge \dfrac{\alpha c}{\epsilon \share_j} -\dfrac{1}{\epsilon \share_i}) \le 2 e^{\dfrac{-2}{\epsilon \alpha} (\dfrac{\alpha c}{\share_j} -\dfrac{1}{\share_i})^2}
\end{align*}
Let $\alpha_j$ be the value that $P(|X_j- {\mathbb E}(X_j)| \ge \dfrac{\alpha_j c}{\epsilon \share_j} -\dfrac{1}{\epsilon \share_i}) \le \dfrac{1}{m n^2}$, then:
\begin{align*}
&2 e^{\dfrac{-2}{\epsilon \alpha_j} (\dfrac{\alpha_j c}{\share_j} -\dfrac{1}{\share_i})^2} \le \dfrac{1}{m n^2}  \\
 &\Rightarrow  2e^{\dfrac{-2}{\epsilon \alpha_j} (\dfrac{\alpha_j c}{\share_j} -\dfrac{1}{\share_i})^2} \le \dfrac{\epsilon}{\alpha_j n^2} \\
  &\Rightarrow \dfrac{2}{\epsilon \alpha_j} (\dfrac{\alpha_j c}{\share_j} -\dfrac{1}{\share_i})^2 \ge ln(\dfrac{2\alpha_j n^2}{\epsilon}) \\
 &\Rightarrow \dfrac{2\alpha_j c^2}{\epsilon \share_j^2}+ \dfrac{2}{\epsilon \alpha_j \share_i^2}-\dfrac{4c}{\epsilon \share_i \share_j} \ge \ln(\dfrac{2\alpha_j n^2}{\epsilon}) \\ 
&\Rightarrow \dfrac{2\alpha_j c^2}{\epsilon \share_j^2}+ \dfrac{2}{\epsilon \alpha_j \share_i^2} - ln(\alpha_j) \ge \dfrac{4c}{\epsilon \share_i \share_j} + \ln(\dfrac{2 n^2}{\epsilon})
\end{align*} 
Since the right hand side of the inequality is a constant, and the left hand side is an increasing function on its domain, we can find an $\alpha'_j$ such that whenever $\alpha_j \ge \alpha'_j$, this inequality holds. Therefore, $\valu_i(B_j) \ge \dfrac{1}{\epsilon \share_i}$ with the probability at least $1-\dfrac{1}{m n^2}$ whenever $\alpha_j \ge \alpha'_j$. Thus, by choosing proper $\alpha$ such that for all ${\agent_j \in \agents}$, $\alpha \ge \alpha'_j$,
 $\WMMS_i$ would not be less than $\dfrac{1}{\epsilon}$ with the probability at least $(1-\dfrac{1}{m n^2})^n \ge 1- \dfrac{1}{m n}$.
\end{proof}

\section{Discussion}
In this work we conduct a study of fair allocation when the agents have different entitlements. The original notion of maxmin share is proposed for the case where all agents are the same in terms of the entitlements. We extend this notion to the case of different entitlement and show that unlike the symmetric case, when the entitlements are different, finding an almost $\WMMS$ allocation is not always possible. More precisely, we show that the best allocation that one can hope for is a $1/n$-$\WMMS$ allocation and that such an allocation can be obtained via a somewhat round-robin procedure.



Our experimental results show that in reality, finding a $\WMMS$ allocation is very likely. Of course this is in contrast to our theoretical results. Therefore, it is important to study the problem under reasonable restrictions that rule out the unlikely worst-case scenarios. We initiate this study by considering two limitations to the problem and show substantially better results for these restricted cases. First, we show that if the valuations of the agents for the items are limited by their maxmin share, one can devise a greedy algorithm to achieve a $1/2$-$\WMMS$ allocation. We then proceed by studying the case where the valuations of the agents for the items are drawn from known distributions. We show that in these cases, a $\WMMS$ allocation exists with high probability. 

Although these observations partially justify our empirical results, it seems that the problem is not yet well understood when it comes to real-world settings. Therefore, we believe that future work can investigate this problem under other reasonable and realistic assumptions and further explain why \textit{almost fair} allocations can be guaranteed in practice.


	
	
\bibliography{santaclaus}
\newpage

\end{document}